%% file: main.tex
\documentclass[a4paper,11pt]{llncs} 

\usepackage[margin=1in]{geometry}

\usepackage{array}
\usepackage{amsmath}
\usepackage{graphics}
\usepackage{epsfig}
\usepackage{latexsym}
\usepackage{amssymb}
\usepackage{color}
\usepackage{url}
\usepackage[margin=1in]{geometry}
\usepackage[ruled,lined,linesnumbered]{algorithm2e}
\usepackage{enumerate}
\usepackage{dsfont}
\usepackage{comment}
\usepackage[all]{xy}
\usepackage{wrapfig}

\input{list-of-notations}

\begin{document}

\title{Computing Zigzag Persistent Cohomology}

 \author{Cl\'ement Maria~\inst{1} \and Steve Oudot~\inst{2}}
\institute{The University of Queensland \email{c.maria@uq.edu.au}
\and INRIA Saclay-\^Ile de France \email{\tt steve.oudot@inria.fr}}

\maketitle

\begin{abstract} 
Zigzag persistent homology is a powerful generalisation of persistent
homology that  allows one not only to compute persistence diagrams with
less noise and using less memory, but also to use persistence in new
fields of application. However, due to the increase in complexity of
the algebraic treatment of the theory, most algorithmic results in the
field have remained of theoretical nature.

This article describes an efficient algorithm to compute zigzag persistence, emphasising on its practical interest. The algorithm is a zigzag persistent cohomology algorithm, based on the dualisation of reflections and transpositions transformations within the zigzag sequence. 

We provide an extensive experimental study of the algorithm. We study the algorithm along two directions. First, we compare its performance with zigzag persistent homology algorithm and show the interest of cohomology in zigzag persistence. Second, we illustrate the interest of zigzag persistence in topological data analysis by comparing it to state of the art methods in the field, specifically optimised algorithm for standard persistent homology and sparse filtrations. We compare the memory and time complexities of the different algorithms, as well as the quality of the output persistence diagrams.

The implementation of of the zigzag persistent cohomology algorithm will be made available as part of the computational library {\tt Gudhi}.
\end{abstract}

\newpage

\section{Introduction}
\label{sec:intro}

{\em Persistent homology} is an algebraic method for tracking
topological features across a growing family of spaces. The theory has
found many applications, in particular in data analysis, where it
allows one to infer the ``shape of data'', {\em i.e.}, to compute
topological properties of a space knowing only a finite set of
sample points. The framework works as follows; starting from a set of
points in a metric space, it builds a nested family of approximations
of the underlying shape, then it tracks the evolution of the
topological features across the family, telling the ones that persist
significantly from the ones that do not---the so-called ``topological
noise''.

Since its introduction fifteen years ago, the theory of persistent
homology has been
well-established~\cite{DBLP:journals/dcg/Cohen-SteinerEH07,DBLP:books/daglib/0025666,DBLP:journals/dcg/EdelsbrunnerLZ02,DBLP:journals/dcg/ZomorodianC05}
and the topic has undergone a deep algorithmic study, with the
introduction of new theoretical
algorithms~\cite{DBLP:journals/comgeo/ChenK13,DBLP:conf/compgeom/Cohen-SteinerEM06,DBLP:journals/dcg/SilvaMV11,DBLP:conf/compgeom/MilosavljevicMS11},
more practical
contributions~\cite{Bauer:arXiv1303.0477,DBLP:conf/alenex/BauerKR14,boissonnatmariacamalgorithmica,Chen11persistenthomology},
and extensions of the original computational
framework~\cite{boissonnatmariaesa2014,DBLP:conf/compgeom/DeyFW14}.
This effort towards better algorithms has in particular led to
dramatic improvements of the running time of computations in practice,
making the memory usage the new bottleneck. From these contributions
have emerged multiple efficient software libraries in the field, {\em
  e.g.}~\cite{dipha_lib,gudhi:PersistentCohomology,dionysus_morozov}.

To address the unavoidable memory limitations in the standard setting of persistent homology, new sparse geometric constructions have been recently defined. They provide faithful approximations to the ``shape of data'' while remaining of moderate size~\cite{DBLP:conf/compgeom/DeyFW14,DBLP:conf/esa/DeySY16,DBLP:conf/isaac/KerberS13,DBLP:journals/dcg/Sheehy13}. These methods have been implemented in software~\cite{simba_lib} and validated experimentally~\cite{DBLP:conf/esa/DeySY16} where, at the cost of small errors in the output persistence diagram, they allow the computation of the persistent homology of significantly larger data sets, which would be otherwise intractable in the standard persistence setting. 

{\em Zigzag persistent homology} is a generalisation of persistent homology that allows to measure and track the topology of spaces that both grow and shrink. This is interesting in practice because it allows one to further reduce the size of the constructions, in addition to reducing the amount of noise in the output diagrams~\cite{os-zz-14}. Unfortunately, the few known algorithms for computing zigzag persistence  prevent us from using the optimisations that have been developed for standard persistence. As a result, the gain in terms of memory usage and quality of the output is counterbalanced by a significant loss in terms of running time and sheer performance.

\vspace{0.3cm} 
\noindent
{\bf Our contributions.} In this article, we 
aim at making a step towards adapting the optimisations developed for
standard persistence to zigzag persistence. For this we introduce a
zigzag persistent {\em cohomology} algorithm by dualising the
reflection and transposition algorithm
of~\cite{DBLP:conf/soda/MariaO15}. While cohomology has been used for
some time now in standard persistence, where it has accelerated the
computations
significantly~\cite{boissonnatmariacamalgorithmica,DBLP:journals/dcg/SilvaMV11,DBLP:conf/compgeom/DeyFW14},
its use in zigzag persistence is still inexistent because it requires
some novel and non-trivial modifications of the approach. This is our
main contribution. 

We also evaluate the resulting gain for the zigzag persistence
computation pipeline in terms of running time and overall
performance. To do so, we run experiments along two directions. First,
we compare the performance of the new zigzag persistent cohomology
algorithm against the existing zigzag persistence
implementation. Second, we provide a qualitative analysis of the
persistence diagrams obtained with the various algorithms in the
field, and we illustrate the interest of zigzag persistence in the
context of topological data analysis.


\section{Background}

\noindent
{\bf Quiver theory.}  
An \emph{$A_n$-type quiver} $\quiv$ is a directed graph:
\[
\xymatrix{\bullet_1 \ar@{<->}[r] & \bullet_2 \ar@{<->}[r] & \cdots \ar@{<->}[r] &
 \bullet_{n-1} \ar@{<->}[r] & \bullet_n\\}
\]
where, by convention in this article, bidirectional arrows are either forward or backward. 

Given a fixed field $(\Field,+,\cdot)$, an
\emph{$\Field$-representation} of $\quiv$ is an assignment of a finite
dimensional $\Field$-vector space $V_i$ for every node $\bullet_i$ and
an assignment of a linear map $f_i\colon V_i \leftrightarrow V_{i+1}$ for every
 arrow $\bullet_i \leftrightarrow \bullet_{i+1}$, the orientation of the map being the same as that of the arrow. We denote such a representation by $\Vmod =
(V_i,f_i)$. In computational topology, an $\Field$-representation of
an $A_n$-type quiver is called a \emph{zigzag module}.


\noindent
\begin{minipage}{0.8\textwidth}
Let $\Vmod = (V_i,f_i)$ and $\Wmod = (W_i, g_i)$ be two
$\Field$-representations of a same quiver $\quiv$.  A \emph{morphism
  of representations} $\phi\colon \Vmod \to \Wmod$ is a set of linear
maps $\{\phi_i: V_i \to W_i\}_{i = 1 \ldots n}$ such that the diagram on the right commutes for every 
arrow of~$\quiv$. The morphism is called an \emph{isomorphism} (denoted by $\cong$) if every 
$\phi_i$ is bijective.  

\end{minipage}
\begin{minipage}{0.2\textwidth}
\[
\xymatrix{V_i \ar@{<->}[r]^{f_i}\ar[d]_{\phi_i} & V_{i+1}\ar[d]^{\phi_{i+1}} \\
          W_{i}\ar@{<->}[r]^{g_i}                       & W_{i+1}}
\]
$\ $
\end{minipage}

The  \emph{direct sum} of two $\Field$-representations $\Vmod = (V_i, f_i),\Wmod = (W_i, g_i)$, denoted by 
$\Vmod \oplus \Wmod$, is the representation of $\quiv$ with space $V_i\oplus 
W_i$ for every node $\bullet_i$, and with map
$f_i \oplus g_i = \left(\begin{smallmatrix}f_i&0\\0&g_i\end{smallmatrix}\right)$ 
for every arrow $\bullet_i \leftrightarrow \bullet_{i+1}$.  An
$\Field$-representation $\Vmod$ is \emph{decomposable} if it can be
written as the direct sum of two non-trivial representations. It is
otherwise said to be \emph{indecomposable}.

Finally, for any $1 \leq b \leq d \leq n$, define the \emph{interval representation} $\Imod[b;d]$ as follows:
\[
\xy
\xymatrix @-0.8pc{
0 \ar@{<->}[r]^-0 & 
\cdots \ar@{<->}[r]^-0 & 
0 \ar@{<->}[r]^-0 & 
\Field \ar@{<->}[r]^-{\idcat} & \cdots \ar@{<->}[r]^-{\idcat} & \Field \ar@{<->}[r]^-0 
& 0 \ar@{<->}[r]^-0 & \cdots \ar@{<->}[r]^-0 & 0 
}
\POS"1,1"."1,3"!C*\frm{_\}},+D*++!U\txt{$\scriptstyle{[1;b-1]}$}
\POS"1,4"."1,6"!C*\frm{_\}},+D*++!U\txt{$\scriptstyle{[b;d]}$}
\POS"1,7"."1,9"!C*\frm{_\}},+D*++!U\txt{$\scriptstyle{[d+1;n]}$}
\endxy
\]
where the maps $0$ and $\idcat$ stand respectively for the null map and the identity map.

Theorem~\ref{thm:zz_gabriel} states every representation of an 
$A_n$-type quiver is decomposable into interval representations, which are the indecomposables for that quiver:

\begin{theorem}[Gabriel~\cite{gabrieltheoremoriginal}]
\label{thm:zz_gabriel}
Every $\Field$-representation $\Vmod$ of an $A_n$-type quiver is
decomposable as a direct sum of indecomposables: $\ \ \ \ \ \Vmod = \Vmod^1 \oplus \Vmod^2 \oplus \cdots \oplus \Vmod^N, \ \ \ \ \ $ 
where each indecomposable $\Vmod^j$ is isomorphic to some interval
representation $\Imod[b_j;d_j]$.
\end{theorem}

For an $\Field$-representation $\Vmod = (V_i, f_i)_{i = 1 \ldots n}$
of an $A_n$-type quiver $\quiv$, and two integers $1\leq b \leq d\leq
n$, we define the {\em restriction} $\Vmod[b;d]$ to be the
representation $(V_i, f_i)_{i = b \ldots d}$ of the quiver
$\quiv[b;d]$ obtained by restricting $\quiv$ to the vertices (and
arrows between them) of indices $b \leq i \leq d$. If $b=1$ then we call
$\Vmod[1;d]$ a \emph{prefix} of $\Vmod$, and if $d=n$ then we call
$\Vmod[b;n]$ a \emph{suffix} of $\Vmod$.  The {\em restriction
principle}~\cite{DBLP:journals/focm/CarlssonS10} states that the
interval decomposition of  $\Vmod[b;d]$ is equal
to the direct sum of the intervals of the decomposition of $\Vmod$
restricted, as intervals, to $[b;d]$.

We finally define two total orders $\leqb$ and $\leqd$ on the
indices $\{1,\ldots ,n\}$ of the vertices of an $A_n$-type quiver,
depending on the orientation of the arrows: given $1\leq i < j\leq n$, 
\begin{align*}
i\leqb j &\ \mbox{if}\ \bullet_{j-1} \to \bullet_j\ \mbox{and}\ i\geqb j\ \mbox{otherwise};\\
i\leqd j &\ \mbox{if}\ \bullet_i \to \bullet_{i+1}\ \mbox{and}\ i\geqd j\ \mbox{otherwise}.
\end{align*}
%
%
For example, in the quiver:
$\ \ \ \ \ \ \xymatrix{\bullet_1 \ar@{->}[r] & \bullet_2 \ar@{->}[r] & \bullet_3 \ar@{<-}[r] & \bullet_4 \ar@{->}[r] & \bullet_5 \ar@{<-}[r] & \bullet_6 \\}\ \ \ \ \ \ $ the indices satisfy $6 \leqb 4 \leqb 1 \leqb 2 \leqb 3 \leqb 5$ and 
$1 \leqd 2 \leqd 4 \leqd 6 \leqd 5 \leqd 3$. 
Note that these orders are a reformulation of the \emph{birth-time}
and \emph{death-time indices}
of~\cite{DBLP:journals/focm/CarlssonS10}. Define $\maxb$ and $\maxd$
to be the maximum functions {\em w.r.t.} the orders $\leqb$ and
$\leqd$ respectively.

\vspace{0.3cm} 
\noindent
{\bf Vector spaces and representations.} Let $V$ be a
finite-dimensional vector space and let $\Basis = \{v_1, \ldots
,v_d\}$ be a basis for $V$. We denote by $\sprod \colon V \times V \to
\Field$ the scalar product associated to $\Basis$, that is: $\langle
v_i,v_j\rangle = \delta_{i,j}$ for any $1\leq i,j\leq d$, where
$\delta_{i,j}$ is the Kronecker symbol. This scalar product makes the
basis orthonormal.

Let $\Vmod = (V_i,f_i)$ be a representation of an $A_n$ type quiver. Let $v$ be a vector of vector space $V_k$. A {\em representative sequence} for $v$ is a collection $(x_1, \ldots, x_n\} \in V_1 \times \ldots \times V_n$ of one vector $x_i$ per vector space $V_i$ of the representation $\Vmod$ such that $v = x_k$ and, either $f_(x_i) = x_{i+1}$ if $\bullet_i \rightarrow \bullet_{i+1}$ is forward, or $x_i = f_i(x_{i+1})$ if $\bullet_i \leftarrow \bullet_{i+1}$ is backward. A representative sequence spans a submodule of $\Vmod$ that is isomorphic to an interval $\Imod[b;d]$, where the indices $\{b, \ldots,d\}$ are exactly the indices for which the $x_i$ are non-zero.

Finally, we 
say that a basis $\{v_1, \ldots , v_d\}$ of a vector space $V_k$ is \emph{compatible}
with the decomposition of $\Vmod$ if there exists, for every vector
$v_j$, a {\em representative sequence} for $v_j$ spanning a submodule that is a summand of $\Vmod$. This interval summand is said to be {\em attached} to $v_j$.

The notion of \emph{compatible} basis has been used in~\cite{DBLP:conf/soda/MariaO15} to maintain algorithmically an interval decomposition with only a vector space basis at a given index $k$.

\vspace{0.3cm} 
\noindent
{\bf Homology, cohomology and persistence theory.} Throughout the paper we use homology and cohomology with coefficients in a fixed field $\Field$, which turns the homology and cohomology groups into $\Field$-vector spaces. In the article, we denote respectively by $\Chain^*(\K)$ and $\Hom^*(\K)$ the direct sum of the cochain groups and the direct sum of the cohomology groups for all dimensions. We refer the reader e.g. to~\cite{Munkres-elementsalgtop1984} for an introduction to homology and cohomology, and to~\cite{DBLP:books/daglib/0025666} for an introduction to persistent homology.

\vspace{0.3cm} 
\noindent
{\bf Zigzag persistence algorithms.} 
There are currently two known approaches to compute zigzag persistent homology. The first one was introduced in~\cite{DBLP:journals/focm/CarlssonS10,DBLP:conf/compgeom/CarlssonSM09} and, similarly to the standard persistent homology algorithm, it maintains and updates a compatible homology basis for the following prefix of the zigzag filtration:

\begin{equation}
\xymatrix @+0.5pc{(\emptyset = \K_{1}) & \K_{2} \ar@{<->}[l] \ar@{<->}[r] 
& \cdots & \K_{i} \ar@{<->}[l]} 
\end{equation}

The homology basis is defined on the complex $\K_i$, and updates are
made under insertion or deletion of a simplex $\sigma$ (depending on
the orientation of the arrow $\K_i \stackrel{\sigma}{\longleftrightarrow}
\K_{i+1}$). This approach has been reinterpreted in terms of a
sequence of matrix multiplications, and its theoretical complexity has
been reduced from cubic to matrix-multiplication time, although this
variant is not
practical~\cite{DBLP:conf/compgeom/MilosavljevicMS11}.

Recently, the authors~\cite{DBLP:conf/soda/MariaO15}  introduced a new algorithm maintaining a compatible basis for the following zigzag filtration:

\vspace{-0.3cm}

\begin{equation}
\xymatrix @-0.1pc{(\emptyset = \K_{1}) & \K_{2} \ar@{<->}[l] \ar@{<->}[r] 
& \cdots & \K_{i} \ar@{<->}[l] & \ar@{<->}[l]_-{\sigma} \K'_{i+1}  
& \cdots \ar@{<->}[l] & \K'_{i+m-1} \ar@{<->}[l]  \ar@{<->}[r] & (\K'_{i+m} = \emptyset)\\ }
\end{equation}

\vspace{-0.3cm}

\noindent
where the first $i$ complexes form the $i$-th prefix of the input
zigzag filtration, and the remaining part of the sequence 
consists of a succession of simplex removals,
in an arbitrary order. Here, the homology basis is not defined at the
end, but on complex $\K_{i}$ where the insertion or removal happens. 
It is maintained under the following two
local transformations of the zigzag filtration (passing from the bottom to the top zigzag, the rest of the zigzag remaining unchanged):

\begin{minipage}{0.45\textwidth}
\begin{equation}\label{eq:filtration_diamond1}
\xymatrix @=20pt @R-2pc{
&& \K\cup\{\sigma\}&&   \\
\cdots \ar@{<->}[r] & \K \ar@{->}[dr]_-{\idcat} \ar@{->}[ur]^-{\sigma} 
&  & \K \ar@{->}[dl]^-{\idcat} \ar@{->}[ul]_-{\sigma} \ar@{<->}[r] & \cdots \\
&&\K&&  \\} \ \ \ \ 
\end{equation}
\end{minipage}
%
%
\begin{minipage}{0.45\textwidth}
\begin{equation}\label{eq:filtration_diamond2}
\xymatrix @=20pt @R-2pc{
&&\K\cup\{\tau\}\\
\cdots \ar@{<->}[r] & \K\cup\{\sigma, \tau\} \ar@{<-}[ur]^-{\sigma} \ar@{<-}[dr]_-{\tau} & & 
\K \ar@{->}[ul]_-{\tau}\ar@{->}[dl]^-{\sigma} \ar@{<->}[r] & \cdots  \\
&&\K\cup\{\sigma\}&&\\}
\end{equation}
\end{minipage}

\bigskip

Here, the homology basis needs to be defined, respectively, on $\K$ and $\K \cup \{\sigma,\tau\}$. In~\cite{DBLP:conf/soda/MariaO15}, the authors show that these transformations are enough to implement a zigzag persistence homology, and that the algorithm shows promising experimental performance. Note that in standard persistent homology, we do not have the freedom of considering different transformations of the filtration, and there is essentially one algorithm in that sense.

In this article, we build upon this last algorithm to construct an efficient zigzag persistent {\em cohomology} implementation.

\vspace{0.3cm} 
\noindent
{\bf Back to quivers: Diamond principles.} Applying the homology functor to the above reflection and transposition diagrams~(\ref{eq:filtration_diamond1}) and~(\ref{eq:filtration_diamond2}), we get the following diagrams of vector spaces and linear maps, called {\em diamonds}:\\
\begin{minipage}{0.45\textwidth}
\begin{equation}\label{eq:representation_diamond1}
\xymatrix @=20pt @R-2pc{
\Wmod := && W &&   \\
\cdots \ar@{<->}[r] & V \ar@{->}[dr]_-{\idcat} \ar@{->}[ur]^-{f} 
&  & V \ar@{->}[dl]^-{\idcat} \ar@{->}[ul]_-{f} \ar@{<->}[r] & \cdots \\
\Vmod := && V &&  \\} \ \ \ \ 
\end{equation}
\end{minipage}
%
%
\begin{minipage}{0.45\textwidth}
\begin{equation}\label{eq:representation_diamond2}
\xymatrix @=20pt @R-2pc{
\Wmod := && W \\
\cdots \ar@{<->}[r] & U' \ar@{<-}[ur]^-{d} \ar@{<-}[dr]_-{c} & & 
U \ar@{->}[ul]_-{b}\ar@{->}[dl]^-{a} \ar@{<->}[r] & \cdots  \\
\Vmod := && V &&\\}
\end{equation}
\end{minipage}\\

\bigskip

\noindent where the vector spaces $U,U',V,W$ and linear maps $f,a,b,c,d$ are homology groups and maps between them, with corank or nullity $1$. The fact that the maps are of corank and nullity $1$ in zigzag persistence~\cite{DBLP:journals/focm/CarlssonS10,DBLP:conf/soda/MariaO15} implies the following properties. The interval decomposition of a zigzag module satisfies, 
\begin{enumerate}
\item[-] for every index $i > 1$, there is at most one interval with birth $i$,
\item[-] for every index $j < n$, there is at most one interval with death $j$.
\end{enumerate}
We assume these properties for the remaining of the article.

The algorithm of~\cite{DBLP:conf/soda/MariaO15} updates the direct sum decomposition of the zigzag module, along with a compatible basis, when passing from the bottom module $\Vmod$ to the top module $\Wmod$ of one of these diamonds. 
The update is driven by several {\em diamond principles}:

\begin{theorem}[Reflection Diamond Principle, sketch~\cite
{DBLP:conf/soda/MariaO15}]
Let $\Vmod$ and $\Wmod$ be respectively the bottom and top zigzag modules in diagram~(\ref{eq:representation_diamond1}), where the reflection happens at quiver index~$i$.

If $f$ is injective of corank $1$, then $\Wmod \cong \Vmod \oplus \Imod[i;i]$. If $f$ is surjective of nullity $1$, then there is a decomposition of $\Vmod \cong \Xmod \oplus \Imod[b;d] \oplus \Imod[b',d'] \oplus \Umod$ s.t. $\Wmod \cong \Xmod \oplus \Imod[b,i-1] \oplus \Imod[i+1,d'] \oplus \Umod'$ where:
\begin{enumerate}
\item $\Imod[b,i-1]$ and $\Imod[i+1,d']$ are a ``cut-out'' version of intervals $\Imod[b;d]$ and $\Imod[b';d']$ (potentially the same interval of the decomposition of $\Vmod$),
\item if $\Umod$ admits $p$ intervals in its decomposition, with birth indices $\{b_1, \ldots, b_p\}$ and death indices $\{d_1,\ldots d_p\}$, then $\Umod'$ admits $p+1$ intervals in its decomposition, with birth indices $\{b', b_1, \ldots, b_p\}$ and death indices $\{d, d_1, \ldots, d_p\}$ paired according to the so-called {\em greedy rule}.
\end{enumerate}
\end{theorem}

We give a full version of the theorem, including a description of the greedy rule (Algorithm~\ref{alg:greedy-rule}), in Appendix~\ref{app:diamonds}. 

In the same spirit, we can define a diamond principle for transpositions~(\ref{eq:representation_diamond2}). This principle includes a case study of the injectivity and surjectivity of maps $a,b,c,d$ and is described in Appendix~\ref{app:diamonds}. 



\section{Diamond Principles for Reversed Reflections and Transpositions}
\label{sec:diamonds}

Applying the {\em cohomology} functor~\cite{Munkres-elementsalgtop1984} to diagrams~(\ref{eq:filtration_diamond1}) and~(\ref{eq:filtration_diamond2}) leads to diagrams of vector spaces (here cohomology groups) where all arrows are reversed:

\begin{minipage}{0.45\textwidth}
\begin{equation}\label{eq:representation_diamond3}
\xymatrix @=20pt @R-2pc{
\Wmod := && W &&   \\
\cdots \ar@{<->}[r] & V \ar@{<-}[dr]_-{\idcat} \ar@{<-}[ur]^-{g} 
&  & V \ar@{<-}[dl]^-{\idcat} \ar@{<-}[ul]_-{g} \ar@{<->}[r] & \cdots \\
\Vmod := && V &&  \\} \ \ \ \ 
\end{equation}
\end{minipage}
%
%
\begin{minipage}{0.45\textwidth}
\begin{equation}\label{eq:representation_diamond4}
\xymatrix @=20pt @R-2pc{
\Wmod := && W \\
\cdots \ar@{<->}[r] & U' \ar@{->}[ur]^-{d} \ar@{->}[dr]_-{c} & & 
U \ar@{<-}[ul]_-{b}\ar@{<-}[dl]^-{a} \ar@{<->}[r] & \cdots  \\
\Vmod := && V &&\\}
\end{equation}
\end{minipage}\\

Diagram~(\ref{eq:representation_diamond4}) is the mirror image of diagram~(\ref{eq:representation_diamond2}) and the same principle for transposition diamonds applies. Differently, we need to introduce a new diamond principle for the \emph{reversed reflection diamond}~(\ref{eq:representation_diamond3}).

\begin{theorem}[Reversed Reflection Diamond Principle, full statement]
\label{thm:reversed_d_princ}
\noindent
\emph{1.} If $g$ is surjective of nullity $1$
in~(\ref{eq:representation_diamond3}), then we have
\[ \Wmod \cong \Vmod \oplus \Imod[i;i]. \]

\noindent
\emph{2.} If $g$ is injective of corank $1$, then let $\{v_1, \ldots,
v_d\}$ be a basis of $V$ that is compatible with the decomposition of
$\Vmod$, let $\langle \cdot,\cdot \rangle$ be the associated scalar
product and denote by $b_j$ and $d_j$ the birth and death of the
interval summand attached to $v_j$. Let $(\im g)^\bot$ be generated by
$\xi \neq 0$. Up to a reordering of the indices, write
\[ \xi = a_1 v_1 + \cdots + a_p v_p \ , \ \ (\im g)^\bot = \langle \xi \rangle \]
with $a_j \neq 0$ for $1 \leq j \leq p$ and $d_1 \geqd \ldots \geqd d_p$. 
Let $b_{\ell_p} = \minb\{b_j\}_{j = 1, \ldots, p}$ and $d_p = \mind\{d_j\}_{j = 1, \ldots, p}$. Then we have
\[
\begin{array}{lcl}
\Vmod\ &\cong& \ \Xmod\ \oplus \ \bigoplus_{1\leq j \leq p} \Imod[b_j;d_j] \ \ \ \ \ \ 
\text{and} \\[5pt]
\Wmod\ &\cong& \ \Xmod\ \oplus\ \Imod[b_{\ell_p};i-1]\ 
                        \oplus\ \Imod[i+1;d_p]\ 
                        \oplus\   \bigoplus_{j=1}^{p-1} \Imod[b_{\ell_j};d_j]\\
\end{array}
\]
where the pairing $(b_{\ell_{j}},d_j)_{1\leq j\leq p-1}$ is computed
according to the following greedy rule (assuming $b_{\ell_p}$ and $d_p$ are considered as already ``paired''):

\vspace{-0.5cm}

\begin{algorithm}
\For{$j$ from $1$ to $p-1$}{
    \lIf{$b_j$ not yet paired}{
      $b_{\ell_j} \leftarrow b_j$; \ \ \ pair $b_{\ell_j}$ with $d_j$}
    \lElse{
      $b_{\ell_j} \leftarrow \displaystyle\minb_{k= 1,\ldots,p \ \ \ } \left\{b_k : b_k \text{ not yet paired}\right\}$;
      \ \ \ pair $b_{\ell_j}$ with $d_j$}
      }
\caption{Pairing for Reversed Reflection Diamond}
\label{alg:greedy-rule_rev}
\end{algorithm}
\end{theorem}

\vspace{-1cm}

\noindent
\begin{minipage}{0.8\textwidth}
\noindent
{\it Proof.} {\bf 1.} 
By surjectivity of $g$ and injectivity of $\idcat_V$, the diamond on the right 
is exact (defined in Appendix~\ref{app:diamonds}). The result follows then from the exact diamond principle and the hypothesis on the interval decomposition made at the beginning of the section. 
\end{minipage}
\begin{minipage}{0.2\textwidth}
\[\xymatrix{
    V \ar[r]^-{\idcat_V} & V\\
    W \ar[u]^-{g} \ar[r]^-{g} & V \ar[u]_-{\idcat_V}\\}\]
\end{minipage}

\noindent
{\bf 2.} We dualise diagram~(\ref{eq:representation_diamond3}), which reverses arrows and dualises vector spaces, while maintaining the same interval decomposition (in the quiver with reversed arrows): 
\begin{equation}
\label{eq:dual_reversed_reflection_diamond}
\xymatrix @R-1pc @-0.7pc{
\Wmod^* \defeq \hspace{-1cm}& & & & \W^* & & & \\
        &\V^*_1 & \cdots \ar@{<->}[l] \ar@{<->}[r] & \V^* \ar@{->}[ru]^f \ar@{->}[rd]_{\idcat} & & \V^* 
       \ar@{->}[lu]_f \ar@{->}[ld]^{\idcat} & \cdots \ar@{<->}[l] \ar@{<->}[r] & \V^*_n \\
\Vmod^* \defeq \hspace{-1cm}& & & & \V^* & & & }
\end{equation}
By duality, $\Vmod \cong \Vmod^*$ and 
$\Wmod \cong \Wmod^*$ and $f: V^* \to W^*$ is the transpose of $g: W\to V$, and the interval decomposition of the primal zigzag module is the same as the one of the dual. In particular, 
$\ker f = \langle \xi^* \rangle$, where $\xi^* = \alpha_1 v^*_1 + \ldots + \alpha_p v^*_p$. The result follows then by applying the reflection diamond principle (Appendix~\ref{app:diamonds}) to this 
dual diagram. Note that orders $\leqb$ and $\leqd$ are reversed by taking the 
dual, because arrows in the corresponding quiver are reversed. \hfill $\Box$



\section{Zigzag Persistent Cohomology Algorithm}
\label{sec:zz_cohomology}
%
Given an input zigzag filtration: $\xymatrix @+0.5pc{(\emptyset =) \ \ \K_{1} & \K_{2} \ar@{<->}[l] \ar@{<->}[r] 
& \cdots & \K_{n-1} \ar@{<->}[l] & \ar@{<->}[l] \K_{n}\\}
$
presented to us through an oracle providing the sequence of simplex insertions and deletions. We compute its persistence. At step $i$ of the algorithm, we maintain a compatible cohomology basis for the persistence module of a zigzag filtration:
\begin{equation}\label{eq:sec_cohomology_filt_input_stepi}
\xymatrix @+0.5pc{(\emptyset =) \ \ \K_{1} \ar@{<->}[r] 
& \cdots & \ar@{<->}[l] \K_{i} 
& \ar@{->}[l]_-{\tau_1} \K'_{i+1} & \ar@{->}[l]_-{\tau_2} \K'_{i+2} 
& \ar@{->}[l]_-{\tau_3} \cdots & \ar@{->}[l]_-{\tau_{m}} \K'_{i+m} \ \ (= \emptyset) \\}
\end{equation}
where the $i$th prefix of the filtrations~(\ref{eq:sec_cohomology_filt_input_stepi}) is identical to the one of the input filtration, and the suffix from indices $i+1$ to $i+m$ in the filtration~(\ref{eq:sec_cohomology_filt_input_stepi}) consists the removal of all the simplices of $\K_i$ one by one, in an arbitrary order; here, the complex $\K_i$ contains exactly $m$ simplices. 

At iteration $i$, the (compatible) cohomology basis is defined on the complex $\K_i$. We represent the basis by an $m \times m$ matrix $\Matrix$, having one column per simplex of $\K_i$, and whose rows form a basis for the cochain groups of all dimensions of $\K_i$. Additionally, we maintain a partition $F \sqcup G \sqcup H$ of the indices $\{1, \ldots, m\}$ and a pairing $G \leftrightarrow H$ satisfying the following conditions:
\begin{itemize}
\item[$\bullet$] the restriction of the cochains $\alpha_\ell$, $\ell \in \{1,\ldots,m\}$, to the simplices of a subcomplex $\K'_{i+j}$ of $\K_i$ is a basis of $\Chain^*(\K'_{i+j})$,
\item[$\bullet$] $\delta \alpha_f = \delta \alpha_h = 0$ for any $f \in F$ and $h \in H$,
\item[$\bullet$] $\delta \alpha_g = \alpha_h$ for any two indices $g \in G$ and $h \in H$that are  paired together.
\end{itemize}
Note that the $\alpha_f$, $f\in F$, are the representatives of a cohomology basis, and the $\alpha_h$, $h \in H$, form a coboundary basis. This partition and associated pairing are already a feature of the cohomology algorithm for standard persistence~\cite{DBLP:journals/dcg/SilvaMV11}, except that the standard persistence implementation can be simplified by maintaining only cochains with an index in $F$. This is not possible in zigzag persistence so far. In the following, we describe the update of the cohomology basis at the level of cochains. Cochain operations like addition translate directly into row operations on the matrix $\Matrix$.

In~\cite{DBLP:conf/soda/MariaO15}, the authors describe how to compute the zigzag persistence of a filtration by applying reflections~(\ref{eq:filtration_diamond1}) and transpositions~(\ref{eq:filtration_diamond2}). In order to compute the zigzag persistent cohomology of a filtration, it is then sufficient to design an algorithm that updates a cohomology basis, which is compatible with the zigzag module, under reflections and transpositions.


\vspace{0.3cm}
\noindent
{\bf Implementation of the zigzag cohomology algorithm.} 
Consider the following operation of the filtration, where the reflection happens at index $i$:
\begin{equation}\label{eq:zz_arrowreflectiondiagram}
\begin{gathered}
\xy
\xymatrix @R-1.3pc @C+0.7pc @-0.6pc{
                     &                            & \K_i \cup \{\sigma\}                    &                 &&\\
\K_{1} \cdots \ar@{<->}[r] & \K_{i} \ar@{->}[dr]_-{\idcat} \ar@{->}[ur]^-{\sigma} 
&  & \K_{i} \ar@{->}[dl]^-{\idcat} \ar@{<-}[r]^-{\tau_1} \ar@{->}[ul]_-{\sigma}& 
   \K'_{i+1} \ar@{<-}[r]^-{\tau_{2}} & \cdots \emptyset \\
&&\K_{i}&&&\\
   }
\endxy
\end{gathered}
\end{equation}
in a zigzag filtration where two copies of $\K_i$ have been inserted at index $i$, which does not change the interval decomposition of the corresponding zigzag module, up to a translation of indices. It induces a reversed reflection diamond at the cohomology level:
\begin{equation}\label{eq:zz_arrowreflectiondiagram_13}
\begin{gathered}
\xy
\xymatrix @R-1.3pc @C+0.6pc @-0.6pc{
        &    & \Hom^*(\K_i \cup \{\sigma\})              &                 &&\\
\Hom^*(\K_{1}) \cdots \ar@{<->}[r] & \Hom^*(\K_{i}) \ar@{<-}[dr]_-{\idcat} 
\ar@{<-}[ur]^-{\sigma^*} 
&  & \Hom^*(\K_{i}) \ar@{<-}[dl]^-{\idcat} \ar@{->}[r]^-{\tau^*_1} \ar@{<-}[ul]_-{\sigma^*}& 
   \Hom^*(\K'_{i+1}) \ar@{->}[r]^-{\tau^*_{2}} & \cdots 0 \\
&&\Hom^*(\K_{i})&&&\\
   }
\endxy
\end{gathered}
\end{equation}
Here, $\sigma^*$ refers to both to the application induced at the cohomology level by the insertion of $\sigma$, and to the cocycle of the cochain group $\Chain^*(\K'_{i+1})$ that satisfies 
$\sigma^*(\sigma) = 1$ and $\sigma^*(\tau) = 0$ otherwise. 
We are given a cohomology basis $\{[\alpha_1], \ldots ,[\alpha_d]\}$ of 
$\Hom^*(\K_i')$ that is compatible with the decomposition of the bottom zigzag module in~(\ref{eq:zz_arrowreflectiondiagram_13}). Define, for any index $\ell \in \{1, \ldots, n\}$, the quantity $c_\ell = \delta \alpha_\ell (\sigma)$.

\vspace{0.3cm}
\indent
{\bf {\it (i) If $\sigma^*$ is surjective}}, then we update the basis to 
$\{[\alpha_1], \ldots ,[\alpha_d], [\sigma^*]\}$ by extending naturally each 
$\alpha_j$ such that $\alpha_j (\sigma) = 0$. 

\vspace{0.3cm}
\indent
{\bf {\it (ii) If $\sigma^*$ is injective}}, then we translate the reverse reflection diamond principle in terms of basis update. Let $b_j$ and $d_j$ be the birth and death of the interval summand attached to $[\alpha_j]$. Suppose, up to a reordering of the indices, that $c_1, \ldots , c_p \neq 0$ and $c_j = 0$ for every $j > p$, and $d_1 \geqd \cdots \geqd d_p$. 
Let $b_{\ell_p} = \minb\{b_j\}_{j = 1, \ldots, p}$ and $d_p = \mind\{d_j\}_{j = 1, \ldots, p}$. We compute the cocycles $\{\alpha'_1, \ldots ,\alpha'_{p-1}, \alpha_{\text{left}}, \alpha_{\text{right}}\}$ defined on $\K_{i} \cup \{\sigma\}$ satisfying:
\begin{itemize}
\item[$\bullet$] $\alpha'_j$ is attached to the interval summand $\Imod[b_{\ell_j};d_j]$ of the top module in~(\ref{eq:zz_arrowreflectiondiagram_13}), 
\item[$\bullet$] $\alpha_{\text{left}}$ and $\alpha_{\text{right}}$ are respectively attached to the interval summands $\Imod[b_{\ell_p};i-1]$ and $\Imod[i+1;d_p]$.
\end{itemize}

We compute the new basis following the algorithm of the reversed reflection diamond principle (Theorem~\ref{thm:reversed_d_princ}): first, we set $\alpha_{\text{right}} \colon = \alpha_p$ and $\alpha_{\text{left}} = \alpha_{\ell_p}$. Second, we apply the procedure:


\begin{algorithm}
\For{$j$ from $1$ to $p-1$}{
    \lIf{$b_j$ not yet paired}{
      set $\alpha'_j \defeq \alpha_j - \frac{c_j}{c_{j_0}} \alpha_{j_0}$, such that 
      $b_{j_0} \leqb b_j$ and $d_{j_0} \leqd d_j$}
    \lElse{
      $b_{\ell_j} \leftarrow \displaystyle\minb_{k= 1,\ldots,p \ \ \ } \left\{b_k : b_k \text{ not yet paired}\right\}$;\ \ \
      $\alpha'_j \defeq \alpha_j - \frac{c_j}{c_{\ell_j}} \alpha_{\ell_j}$
        }
      }
\caption{Pairing for Reversed Reflection Diamond}
\label{alg:greedy-rule-reversed-diamond}
\end{algorithm}

\vspace{0.3cm}

\noindent
{\bf Transposition diamonds.} The basis update for the transposition diamond is simpler, and is explicitly described in the statement of the Transposition Diamond Principle in Appendix~\ref{app:diamonds}. The basis updates follows directly the update of vectors $u$, $v$ and $v+ \gamma \cdot u$ in the case study of the theorem. This is merely a dualized version of the basis update for transposition diamonds introduced in~\cite{DBLP:conf/soda/MariaO15}.

\vspace{0.3cm}
\noindent
{\bf Correction of the algorithm.} Note that the update described in the case $\sigma^*$ surjective directly leads to a cohomology basis of $\Hom^*(\K_i \cup \{\sigma\})$ that is compatible with the interval decomposition of the top zigzag module in~(\ref{eq:zz_arrowreflectiondiagram_13}). This comes straight from the fact that, by surjectivity of $\sigma^*$, the only representative sequence for $[\sigma^*]$ (up to multiplication by a scalar) is $(0, \ldots, 0, [\sigma^*], 0, \ldots, 0)$, and it spans a summand isomorphic to $\Imod[i;i]$. By virtue of Theorem~\ref{thm:reversed_d_princ}, this fulfills the decomposition of the top module of diagram~(\ref{eq:zz_arrowreflectiondiagram_13}).

The case $\sigma^*$ injective requires more work. We start by introducing basic properties of the cohomology basis of a complex under an elementary inclusion.

\begin{lemma}
\label{lem:elementary_prop_inclusion_cohomology}
Let $\xymatrix{\K \ar[r]^-\sigma & \K \cup \{\sigma\}}$ be an elementary 
inclusion and let $g: \Hom^*(\K \cup \{\sigma\}) \to \Hom^*(\K)$ be the morphism 
induced at the cohomology level by this elementary inclusion. 
Let $\{ [\alpha_i] \}_{i \in I}$ be any basis for 
$\Hom^*(\K)$ and $\langle \cdot,\cdot \rangle$ be the associated scalar product. 
Define, for every index $i$, $c_i = \delta \alpha_i (\sigma) = \alpha_i(\partial \sigma)$. 
\begin{enumerate}[(i)]
\item if all $c_i = 0$, then $g$ is surjective of nullity $1$, otherwise $g$ is injective of corank $1$,
\item if there is a $c_{i_0} \neq 0$, then any cohomology class $[\alpha_j - \frac{c_j}{c_{i_0}} \alpha_{i_0}]$ belongs to $\im g$, for $j \neq i_0$,
\item if $g$ is injective, then $(\im g)^\bot$ is generated by the cohomology class $[\sum_{j \in I}\alpha_j - \frac{c_j}{c_{i_0}} \alpha_{i_0}]$.
\end{enumerate}
\end{lemma}

In light of this lemma, we see that Algorithm~\ref{alg:greedy-rule-reversed-diamond} reduces the cohomology basis $\{[\alpha_f]\}_{f \in F}$ so that it can be partitioned into a basis of $\im g$ and a basis of $(\im g)^\bot$; the order in which the reduction is computed coincides with the new decomposition of the top zigzag module dictated by the Reversed Reflection Diamond Principle. This is detailed at the end of the section.

Note that Lemma~\ref{lem:elementary_prop_inclusion_cohomology}(i) gives a computable criterion to check whether a map $\sigma^*$ is injective or surjective. Note also that items (i) and (ii) are already used, with a specific cohomology basis, in the cohomology algorithm for standard persistence~\cite{DBLP:journals/dcg/SilvaMV11}. 

\begin{proof}[Lemma~\ref{lem:elementary_prop_inclusion_cohomology}] 
\noindent
{\bf {\bf (i)}} The proof of~\cite{DBLP:journals/dcg/SilvaMV11}, relying on a specific cohomology basis, adapts directly to our lemma.


For the remainder of the proof, we define for every cocycle $\alpha_i$, $i \in I$, of $\Chain^*(\K)$ the cochain $\alpha_i \in \Chain^*(\K \cup \{\sigma\})$, 
extended to the simplices of $\K \cup \{\sigma\}$ by setting $\alpha_i(\sigma) = 0$. We use the same notation $\alpha_i$. This is not a cocycle in general because $\delta \alpha_i (\sigma) = c_i$ may not be~$0$. 
Note that, for an application $g$ induced at the cohomology level by an elementary inclusion, $g([\alpha'])$ is equal to $[\alpha] \in \Hom^*(\K)$, where $\alpha'$ is a cocycle of $\Chain^*(\K \cup \{\sigma\})$ and $\alpha$ is the cocycle of $\Chain^*(\K)$ equal to the restriction of $\alpha'$ to the simplices of $\K$---see~\cite{Munkres-elementsalgtop1984}.

\vspace{0.2cm}

\noindent
{\bf (ii)} If $c_{i_0} \neq 0$, then every $\alpha_j - \frac{c_j}{c_{i_0}} \alpha_{i_0}$, for $j \in I \setminus \{i_0\}$, is a cocycle of $\Chain^*(\K \cup \{\sigma\})$. Hence, every $[\alpha_j - \frac{c_j}{c_{i_0}} \alpha_{i_0}]$ is a cohomology class of $\Hom^*(\K)$ belonging to $\im g$.

\vspace{0.2cm}

\noindent
{\bf (iii)} The family $\{[\alpha_i - \frac{c_i}{c_{i_0}} \alpha_{i_0}]\}_{i \in I \setminus \{i_0\}}$ is a basis of $\im g$, because the classes $\{[\alpha_i]\}_{i \in I}$ are linearly independent and $\dim (\im g) = |I| - 1$. 
%
The cohomology class $[\sum_{i \in I} c_i \alpha_i]$ is non zero and orthogonal to $\im g$, as for any $j \in I\setminus \{i_0\}$, $\langle [\sum_{i \in I} c_i \alpha_i], [\alpha_i - \frac{c_i}{c_{i_0}} \alpha_{i_0}] \rangle = 0$.            \hfill $\Box$
\end{proof}

Now, we prove that Algorithm~\ref{alg:greedy-rule-reversed-diamond} is correct and terminates, with the following lemma regarding the pairing of birth and death indices in the algorithm:

\begin{lemma}
\label{lem:lemma_indices_pairing_algorithm}
At the $j$th iteration of the algorithm, the following is true:
\begin{enumerate}[(i)]
\item if ``$b_j$ not yet paired'', then there exists an index $j_0$ satisfying $b_{j_0} \leqb b_j$ and $d_{j_0} \leqd d_j$,
\item if ``$b_j$ already paired'', then $b_{\ell_j} \geqb b_j$.
\end{enumerate}
\end{lemma}

\begin{proof} 
{\bf (i)} Note first that $\ell_p \neq j$ because $b_j$ is not yet paired at iteration $j$, and $b_{\ell_p}$ is marked as ``paired'' before executing Algorithm~\ref{alg:greedy-rule_rev}. 
We study the different cases. 

If $p = \ell_p$, then index $j_0 = p = \ell_p$ has the desired properties because $b_{\ell_p}$ and $d_p$ are minimal in their respective orders. Suppose then $p \neq \ell_p$. 

If $\ell_p > j$, then $d_{\ell_p} \leqd d_p$ because death indices are sorted, and so $j_0 = \ell_p$ has the desired properties. 

If $\ell_p < j$, then the set of indices $k$ for which $d_k$ is
paired with a different birth index, \emph{i.e.} $R = \{k : k < j
\ \mbox{and}\ b_{\ell_k} \neq b_k \}$, is non-empty (it contains in
particular $\ell_p$). Let $k_0$ be maximal in this set. Then,
$b_{\ell_{k_0}}$ is minimal in the order $\leqb$, among
births that are not yet paired at iteration $k_0$, which implies
$b_{\ell_{k_0}} \leqb b_j$ and $\ell_{k_0} > k_0$. Finally,
$\ell_{k_0} > j$ by maximality of $k_0$ in the set $R$. Consequently,
$d_{\ell_{k_0}} \leqd d_j$, and so $j_0 = \ell_{k_0}$ has the
desired properties.

\vspace{0.2cm}
\noindent
{\bf (ii)} The property follows from the fact that $\minb\left\{b_k~:~1 \leq k \leq p \ \text{and} \ b_k \text{ not yet paired}\right\}$ only increases in the order $\leqb$ during the execution of the algorithm. \hfill $\Box$
\end{proof}

Lemma~\ref{lem:lemma_indices_pairing_algorithm}(i) ensures that there is always an index $j_0$ to pick when executing line $2$ of Algorithm~\ref{alg:greedy-rule-reversed-diamond}. As mentioned earlier, Algorithm~\ref{alg:greedy-rule-reversed-diamond} reduces the cohomology basis so that it aligns with the direct sum decomposition $\Hom^*(\K) = \im g \oplus (\im g)^\bot$. We finally get a basis of $\Hom^*(\K_i \cup \{\sigma\})$ by discarding the cohomology class of the reduced basis generating $(\im g)^\bot$, and by extending all obtained cocycles $\alpha$ with $\alpha(\sigma) = 0$.

Finally, we conclude:

\begin{lemma}
The basis of $\Hom^*(\K \cup \{\sigma\})$ obtained with
Algorithm~\ref{alg:greedy-rule-reversed-diamond} is compatible with
the top zigzag module in~(\ref{eq:zz_arrowreflectiondiagram_13}).
\end{lemma}

\begin{proof}
Given two vectors $v$ and $w$ in $V_i$, and two representative sequences $(x_1, \ldots, v, \ldots , x_{i+m})$ and $(y_1, \ldots, w, \ldots , y_{i+m})$, we can extend additions and scalar multiplications in $V_i$ to the representative sequences:
\[
v + c \cdot w \ \text{is represented by}\ (x_1 + c \cdot y_1, \ldots, v + c \cdot w, \ldots , x_{i+m} + c \cdot y_{i+m})
\]
By the definition of orders $\leqb$ and $\leqd$, if the representatives sequences span submodules isomorphic respectively to $\Imod[b,d]$ and $\Imod[b';d']$, their sum, provided $c \neq 0$, spans a submodule isomorphic to $\Imod[\maxb \{b,b'\}; \maxd \{d;d'\}]$. Using Lemma~\ref{lem:lemma_indices_pairing_algorithm}(ii) and the fact that all arrows from indices $i+1$ to $i+m$ are backward, we conclude that at iteration $j$ of Algorithm~\ref{alg:greedy-rule-reversed-diamond}, the element $\alpha'_j$ produced represents a cohomology class attached to the interval summand $\Imod[b_{\ell_j}; d_j]$. Note that, since we are in cohomology, orders $\leqb$ and $\leqd$ are reversed when dualizing. We conclude using Theorem~\ref{thm:reversed_d_princ}. \hfill $\Box$
\end{proof}



\noindent
{\bf Remark:} For the sake of simplicity, we have not provided the update procedure for cochain $\alpha_g$, $g \in G$ or $\alpha_h$, $h \in H$, in the main description of the implementation. To maintain the pairing $g \leftrightarrow h$, standing for $\delta \alpha_g = \alpha_h$, under the insertion of $\sigma$, we set $\alpha_h(\sigma) = \alpha_g(\partial \sigma)$.

Note that this is not necessary in the cohomology algorithm for standard persistence~\cite{DBLP:journals/dcg/SilvaMV11}, as only the set of cocycles $\{\alpha_f\}_{f \in F}$ is maintained. In all known algorithms for zigzag persistence, the sets of cochains $\{\alpha_g\}_{g \in G}$ and $\{\alpha_h\}_{h \in H}$ need to be maintained to process backward arrows.

\begin{figure}[t]
\centering
\hspace{1.3cm} Standard persistence \hspace{0.2cm} Sparse persistence \hspace{1cm} Zigzag Persistence 
\setlength{\tabcolsep}{5pt}
\begin{tabular}{|l|rr|rr|rr|rrr|}
\hline
Data & \#$P$ & $d$ & \#$\K_{\text{max}}$ & $T_{\text{std pers.}}$ & \#$\K_{\text{max}}$ & $T_{\text{sparse pers.}}$ & \# arrows & \#$\K_{\text{max}}$ & $T_{\text{zz pers.}}$ \\
\hline
{\bf Cli} & $2000$ & $4$ & $230 \cdot 10^6$ & $3147$ sec. & $7474$ &  $24$ sec. & $2.4 \cdot 10^6$ & $50840$ & $285$ sec. \\ 
\hline
\end{tabular}
\setlength{\tabcolsep}{2pt}
\begin{tabular}{l cc c cc r}
\begin{minipage}{0.3\textwidth}
\includegraphics{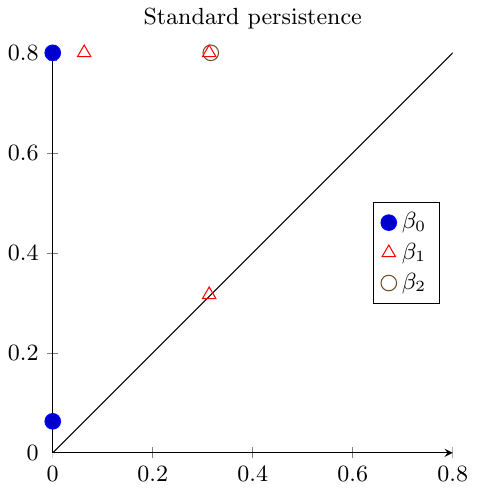}
\end{minipage} 
& & & 
\begin{minipage}{0.3\textwidth}
\includegraphics{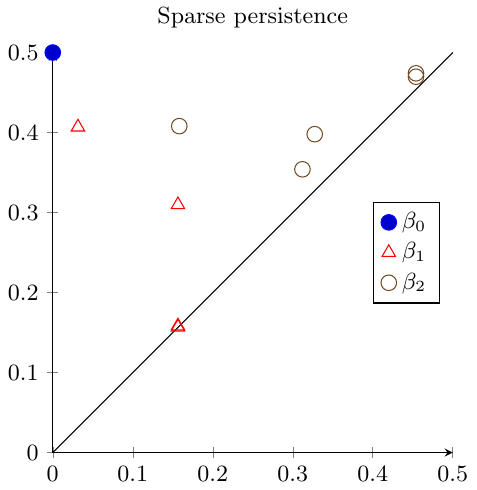}
\end{minipage} 
& & & 
\begin{minipage}{0.3\textwidth}
\includegraphics{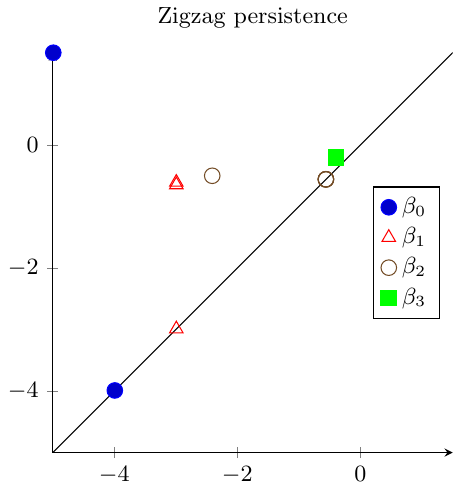}
\end{minipage} \\
\end{tabular}
\caption{Performance and best possible persistence diagrams obtained by the different methods on the data {\bf Cli}. Note that the zigzag persistence diagram on the right has two superimposed $1$-dimensional persistent features.}
\label{fig:cli_diagrams}
\end{figure}

\section{Experiments}
\label{sec:experiments}

In this section we report on the performance of the zigzag persistent cohomology algorithm presented above. The objective of the experiments is twofold. First, we aim at positioning zigzag persistence in the scope of computational methods in topological data analysis. To do so, we chose to compare the performance of zigzag persistence with state of the art methods in the field; specifically, optimised algorithm for standard persistence and persistence of sparse filtration. 
We find that zigzag persistence capable of extracting more topological information from point clouds, where the other two methods fail. 
Second, we aim at comparing the time and memory performance of the zigzag persistent cohomology algorithm introduced in this article and zigzag persistent homology. We find that the zigzag persistent cohomology algorithm scales significantly better to more challenging examples. 

We run experiments on a $3.40$GHz computer with $16$GiB RAM, running Linux. All codes are compiled with {\tt gcc 5.4.0}. Timings are measured by the ${\tt clock()}$ function and memory consumption is measured by the ${\tt getrusage()}$ function in {\tt C}. We use a collection of datasets, both generated and from real life measurements. The generated datasets are {\bf Cli}, containing points from the {\em Clifford dataset}, described in~\cite{os-zz-14}, which samples different manifolds at different scales, and {\bf S3\_r}, which contains points sampled uniform randomly on the unit $3$-sphere embedded in $\R^4$. The ``real life'' datasets are {\bf NatI}, coming from high-contrast patches extracted from natural images~\cite{Carlsson:2008:LBS:1325290.1325300}, and {\bf GePh}, coming from the Gesture Phase Segmentation~\cite{Lichman:2013,DBLP:journals/eswa/MadeoPL16}. Figures~\ref{fig:perfzz} give details about the size \#P and embedding dimension $d$ of the point sets. They also provide the maximal size of a complex stored ``\#$\K_{\text{max}}$'' and the total number of arrows ``\# arrows'' encountered (in zigzag persistence) during computation.

All filtered complexes in our experiments are {\em Rips filtrations}~\cite{DBLP:books/daglib/0025666} for standard persistence, {\em Simplicial batch collapse filtrations}~\cite{DBLP:conf/esa/DeySY16} for sparse persistence, and {\em oscillating Rips zigzag filtrations}~\cite{os-zz-14} for zigzag persistence. These construction are known, both theoretically and practically, to furnish correct persistence diagrams.


\vspace{0.3cm}
\noindent
{\bf Zigzag persistence and standard persistence.} In this section, we compare the performance of zigzag persistence with existing state of the art methods in topological data analysis. We compare our implementation with the standard persistence algorithm of the library {\tt Gudhi}~\cite{gudhi:PersistentCohomology}, which is reported as one of the state of the art persistence algorithms in the field\footnote{The objective of this section is not to compare the efficiency of computational libraries for standard persistence, but rather the ability of different methods to extract topological features from data. We refer to~\cite{DBLP:journals/corr/OtterPTGH15} for one of the detailed studied of the different software libraries for standard persistence.}. We also compare zigzag persistence to the more recent persistent homology of sparse filtrations implemented in the software {\tt Simba}~\cite{simba_lib}, which shows a significantly better memory performance in practice~\cite{DBLP:conf/esa/DeySY16}.

We use the dataset {\bf Cli} with Rips threshold $\rho = 0.8$ for standard persistence, approximation ratio $c=1.001$ for sparse persistence and multipliers $\eta = 2.7$ and $\rho = 2.75$ for the oscillating Rips zigzag persistence. This input is of interest as the points sample a torus, embedded in a $3$-sphere. Hence, at small scales the data represent a torus, and a (sparsely) sampled $3$-sphere at larger scales. Both standard and sparse persistence capture the torus well, but, when standard persistence with {\tt Gudhi} requires a complex of hundred of million simplices for the computation, the sparse persistence of {\tt SimBa} shows much better performance in both time and space, at a cost of a slightly noisier diagram. Only zigzag persistence captures the $3$-dimensional feature that, despite being close to the diagonal compare to the toric features, is the only point in dimension $3$ (no topological noise) and hence can be read as an existing topological feature. Additionally, the zigzag diagram is globally less noisy. During the experiments, the standard persistence was rapidly limited by memory when attempting to extract a $3$-dimensional feature, and the sparse persistence could not find one, even for smaller approximation ratios $c$.

The ability of zigzag persistence to extract topological features of sparsely sampled spaces is remarkable, as it is a common situation in topological data analysis. Note that for our implementation of zigzag persistence, the number of updates, the maximal size of a complex stored and the overall timing remain small.


\vspace{0.3cm}
\noindent
{\bf Performance of zigzag persistence algorithms.} In the section, we compare the time and memory performance of our zigzag persistent cohomology algorithm with the zigzag persistent homology algorithm of the software {\tt Dionysus}~\cite{dionysus_morozov}, based on the theoretical work~\cite{DBLP:journals/focm/CarlssonS10,DBLP:conf/compgeom/CarlssonSM09}. We use generated datasets, with known ground truth topology, and datasets from real life measurements.

\begin{figure}[t]
\centering

\setlength{\tabcolsep}{5pt}
\begin{tabular}{|l|rrrrrr|rr|rr|}
\hline
Data & \#$P$ & $d$ & $\eta$ & $\rho$ & \#arrows & \#$\K_{\text{max}}$ & $T_{\text{{\tt Dio}}}$ & $M_{\text{{\tt Dio}}}$ & $T_{\text{{\tt zz-coH}}}$ & $M_{\text{{\tt zz-coH}}}$ \\
\hline
{\bf Cli} & $2000$ & $4$ & $4.00$ & $4.01$ & $17 \cdot 10^6$ & $1557194$ & $16209$ sec. & $4335$ MB. & $644$ sec. & $1008$ MB. \\
{\bf NatI} & $595$ & $25$ & $3.1$ & $3.11$ & $23 \cdot 10^6$ & $983838$ & $2931$ sec. & $1438$ MB. & $277$ sec. & $620$ MB. \\
{\bf S3\_r} & $1000$ & $4$ & $3.2$ & $3.21$ & $44 \cdot 10^6$ & $1649890$ & $13574$ sec. & $2969$ MB. & $2969$ sec. & $1226$ MB. \\
{\bf GePh} & $1747$ & $19$ & $14$ & $14.1$ & $25 \cdot 10^6$ & $357747$ & $12153$ sec. & $7251$ MB. & $186$ sec. & $364$ MB. \\
\hline
\end{tabular}
\caption{Time and memory performance of the zigzag persistent cohomology algorithm introduced in this article and the zigzag persistent homology algorithm of {\tt Dionysus}.}
\label{fig:perfzz}
\end{figure}

Figure~\ref{fig:perfzz} shows the time ($T$) and memory ($M$) performance of the software {\tt Dionysus} ({\tt Dio}) and our zigzag persistent cohomology algorithm ({\tt zz-coH}) on large instances. Our implementation outperforms significantly {\tt Dionysus} on both time and memory, with running time up to $25$ times faster and memory consumption up to $4.3$ times smaller.


\vspace{0.3cm}
\noindent
{\bf Conclusion.} In light of this study, zigzag persistence completes the toolbox of computational methods in topological data analysis, and is particularly useful for sparsely sampled spaces. In a nutshell, standard persistence is the method of choice for well-sampled spaces and sparse persistence allows the computation to scale up to much larger datasets. Zigzag persistence allows to extract topological features from data of worse quality, while keeping the size of complexes small, but running slower, in general, than sparse persistence. 
Finally, the zigzag cohomology algorithm introduced in this article scales significantly better than the existing zigzag algorithm to larger and more complex inputs. This is an encouraging step towards optimising zigzag persistence for practical purpose, in the same spirit as the prolific and successful research for the optimisation of standard persistence.

The implementation of the zigzag persistent cohomology will be available as part of the software library {\tt Gudhi} at next release.

\newpage

\bibliographystyle{plain}
\bibliography{bibliography}

\newpage

\appendix


\section{Diamond Principles}
\label{app:diamonds}

\noindent
{\bf Exact diamonds.} Consider the diagram:

\begin{equation}
\label{eq:zz_exactdiamonddiagram}
\xymatrix @R-1pc @-0.9pc{
\Wmod\defeq \hspace{-1cm}& & & & \W_i & & & \\
        &\V_1 & \cdots \ar@{<->}[l] \ar@{<->}[r] & \V_{i-1} \ar@{->}[ru]^-{b} \ar@{<-}[rd]_-{a} & & \V_{i+1} 
       \ar@{->}[lu]_-{d} \ar@{<-}[ld]^{c} & \cdots \ar@{<->}[l] \ar@{<->}[r] & \V_n \\
\Vmod\defeq \hspace{-1cm}& & & & \V_i & & & }
\end{equation}
We say that the following diagram:
\begin{equation}
\label{eq:zz_exactness}
\xymatrix{
    V_{i+1} \ar[r]^-{d} & W_i\\
    V_i \ar[u]^-{c} \ar[r]^-{a} & V_{i-1} \ar[u]_-{b}\\}
\end{equation}
is \emph{exact}~\cite{DBLP:journals/focm/CarlssonS10} if $\im D_1 = \ker D_2$ in the sequence 
$\xymatrix{V_i \ar@{->}[r]|-{D_1} & V_{i-1} \oplus V_{i+1} \ar@{->}[r]|-{D_2} & W_i}$, 
where $D_1(v) = (a(v),c(v))$ and $D_2(x,y) = b(x) - d(y)$. Note in particular that an exact diamond 
commutes, i.e. $b \circ a = d \circ c$.

If diagram~(\ref{eq:zz_exactness}) is exact we say that the representations $\Vmod$ and $\Wmod$ are 
related by an \emph{exact diamond at index $i$}. We recall the \emph{Exact Diamond Principle}:
\begin{theorem}[Exact Diamond Principle~\cite{DBLP:journals/focm/CarlssonS10}]
\label{thm:zz_exactdiamondprinciple}
Given $\Vmod$ and $\Wmod$ related by an exact diamond at index $i$, there is a partial bijection between 
the intervals of the decompositions of $\Vmod$ and $\Wmod$:
\begin{enumerate}
\item[-] intervals $\Imod[i;i]$ are unmatched,
\item[-] for $b<i$, intervals $\Imod[b;i]$ are matched with intervals $\Imod[b;i-1]$ and vice versa,
\item[-] for $d>i$, intervals $\Imod[i;d]$ are matched with intervals $\Imod[i+1;d]$ and vice versa,
\item[-] intervals $\Imod[b;d]$ are matched with intervals $\Imod[b;d]$ in all other cases.
\end{enumerate}
\end{theorem}

\vspace{0.3cm}
\noindent
{\bf Reflection diamond principle.} 
%
The reflection diamond principle relates the interval decomposition of the bottom and top zigzag modules of:
\begin{equation}
\label{eq:zz_weakdiamonddiagram}
\xymatrix @R-1pc @-0.7pc{
\Wmod\defeq \hspace{-1cm}& & & & \W & & & \\
        &\V_1 & \cdots \ar@{<->}[l] \ar@{<->}[r] & \V \ar@{->}[ru]^f \ar@{->}[rd]_{\idcat} & & \V 
       \ar@{<-}[lu]_f \ar@{<-}[ld]^{\idcat} & \cdots \ar@{<->}[l] \ar@{<->}[r] & \V_n \\
\Vmod\defeq \hspace{-1cm}& & & & \V & & & }
\end{equation}

\begin{theorem}[Reflection Diamond Principle, full statement~\cite{DBLP:conf/soda/MariaO15}]
\noindent
\emph{1.} if $f$ is injective of corank $1$, we have
\[ \Wmod \cong \Vmod \oplus \Imod[i;i] \]

\noindent
\emph{2.} if $f$ is surjective of nullity $1$. Let $\{v_1, \ldots, v_d\}$ be a basis of $V$ that is compatible with the decomposition of $\Vmod$ 
and denote by $b_j$ and $d_j$ the birth and death of the interval summand attached to $v_j$. Let $\ker g$ be generated by $\xi \neq 0$. Up to a reordering of the indices, write 
\[ \xi = a_1 v_1 + \cdots + a_p v_p \ , \ \ \ker g = \langle \xi \rangle \]
with $a_j \neq 0$ for $1 \leq j \leq p$ and $d_1 \leqd \ldots \leqd d_p$. 
Let $b_{\ell_p} = \maxb\{b_j\}_{j = 1, \ldots, p}$ and $d_p = \maxd\{d_j\}_{j = 1, \ldots, p}$, we have
\[
\begin{array}{lcl}
\Vmod\ &\cong& \ \Umod\ \oplus\ \bigoplus_{1\leq j \leq p} \Imod[b_j;d_j] \ \ \ \ \ \ 
\text{and} \\[5pt]
\Wmod\ &\cong& \ \Umod\ \oplus\ \Imod[b_{\ell_p};i-1]\ \oplus\ 
\Imod[i+1;d_p]\ \oplus\ \bigoplus_{j=1}^{p-1} \Imod[b_{\ell_j};d_j]\\
\end{array}
\]
where the pairing $(b_{\ell_{j}},d_j)_{1\leq j\leq p-1}$ is computed
as follows (assuming $b_{\ell_p}$ and $d_p$ are considered as already ``paired''):

\begin{algorithm}
\For{$j$ from $1$ to $p-1$}{
    \lIf{$b_j$ not yet paired}{
      $b_{\ell_j} \leftarrow b_j$; \ \ \ pair $b_{\ell_j}$ with $d_j$}
    \lElse{
      $b_{\ell_j} \leftarrow \displaystyle\maxb_{k= 1,\ldots,p \ \ \ } \left\{b_k : b_k \text{ not yet paired}\right\}$;
      pair $b_{\ell_j}$ with $d_j$}
      }
\caption{Pairing for Surjective Diamond}
\label{alg:greedy-rule}
\end{algorithm}
\end{theorem}

\noindent
{\bf Transposition diamond principle.} Consider the diagram:
\begin{equation}
\label{eq:zz_transpositiondiamonddiagram}
\xymatrix @R-1pc @-0.9pc{
\Wmod\defeq \hspace{-1cm}& & & & \W_i & & & \\
        &\V_1 & \cdots \ar@{<->}[l] \ar@{<->}[r] & \V_{i-1} \ar@{->}[ru]^-{b} \ar@{->}[rd]_-{a} & & \V_{i+1} 
       \ar@{<-}[lu]_-{d} \ar@{<-}[ld]^{c} & \cdots \ar@{<->}[l] \ar@{<->}[r] & \V_n \\
\Vmod\defeq \hspace{-1cm}& & & & \V_i & & & }
\end{equation}
We say that the representations $\Vmod$ and $\Wmod$ are related by a \emph{transposition diamond} if 
the following diagram is \emph{exact}:

\begin{figure}[h!]
\begin{minipage}{0.35\linewidth}
$\xymatrix{ W_i \ar[r]^-{d} & V_{i+1}\\
    V_{i-1} \ar[u]^-{b} \ar[r]^-{a} & V_i \ar[u]_-{c}\\}$ 
\end{minipage}
\begin{minipage}{0.60\linewidth} 
$\xymatrix{V_{i-1} \ar@{->}[r]|-{D_1} & V_{i} \oplus W_{i} \ar@{->}[r]|-{D_2} & V_{i+1}}$
with $D_1(v) = (a(v), b(v))$ and $D_2(x,y) = c(x) - d(y)$ such that $\im D_1 = \ker D_2$
\end{minipage}
\end{figure}
Note that the transposition diamond diagram~(\ref{eq:zz_transpositiondiamonddiagram}) is similar to 
the exact diamond diagram~(\ref{eq:zz_exactdiamonddiagram}) except that the diamond is ``rotated by 
$90^o$''.

\begin{theorem}[Transposition Diamond Principle~\cite{DBLP:conf/soda/MariaO15}]
\label{thm:zz_transpositionprinciple}
Given $\Vmod$ and $\Wmod$ related by a transposition diamond as above, we assume that the 
maps $a,b,c,d$ are of two different types: injective of corank $1$ and surjective of nullity $1$. 
We have:

\noindent
{\bf 1.} if $a$ and $c$ surjective of nullity $1$ then $\Vmod \cong \Umod \oplus \Imod[b;i-1] 
\oplus \Imod[b';i]$ for some indices $b,b' \leq i-1$. 
Let $(\ldots, u, 0,0 \ldots)$ and $(\ldots v, a(v), 0 \ldots)$, $u,v \in V_{i-1}$, 
be representative 
sequences for the interval summands $\Imod[b;i-1]$ and $\Imod[b';i]$ respectively. There exists 
$\gamma \in \Field$ such that $v + \gamma u \in \ker b$ and:
  \begin{enumerate}
  \item[(i)] if $\gamma = 0$ then $\Wmod \cong \Umod \oplus \Imod[b;i] \oplus \Imod[b';i-1]$,
  \item[(ii)] if $\gamma \neq 0$ then $\Wmod \cong \Umod \oplus \Imod[\maxb \{b,b'\};i-1]\oplus \Imod[\minb\{b,b'\};i]$.
  \end{enumerate}

\noindent
{\bf 2.} if $a$ and $c$ injective of corank $1$ then $\Vmod \cong \Umod \oplus \Imod[i;d] \oplus 
\Imod[i+1;d']$ for some indices $d,d' \geq i+1$. 
Let $(0 \ldots 0,v,c(v),\ldots)$ and $(0 \ldots 0, 0, u, \ldots)$, $v \in V_i$ and 
$u \in V_{i+1}$, be representative sequences for the interval summands $\Imod[i;d]$ and 
$\Imod[i+1;d']$ respectively. There exists $\gamma \in \Field$ such that $u + \gamma c(v) \in \im d$ and:
  \begin{enumerate}
  \item[(i)] if $\gamma = 0$ then $\Wmod \cong \Umod \oplus \Imod[i+1;d] \oplus \Imod[i;d']$,

\vspace{0.2cm}

  \item[(ii)] if $\gamma \neq 0$ then $\Wmod \cong \Umod \oplus \Imod[i; \maxd\{d,d'\}] \oplus \Imod[i+1; \mind \{d,d'\}]$.
  \end{enumerate}

\noindent
{\bf 3.} if $a$ injective of corank $1$ and $c$ surjective of nullity $1$ then:\\
$$\Vmod \cong \Umod \oplus \Imod[i;d] \oplus \Imod[b;i] \ \text{  and  } \  
\Wmod \cong \Umod \oplus \Imod[i+1;d] \oplus \Imod[b;i-1]$$

\noindent
{\bf 4.} if $a$ surjective of nullity $1$ and $c$ injective of corank $1$ then:\\
$$\Vmod \cong \Umod \oplus \Imod[i+1;d] \oplus \Imod[b;i-1]  \ \text{  and  } \  
\Wmod \cong \Umod \oplus \Imod[i;d] \oplus \Imod[b;i]$$
\end{theorem}

\end{document}

%% file: list-of-notations.tex
\newcommand*{\defeq}{\mathrel{\vcenter{\baselineskip0.5ex \lineskiplimit0pt
                     \hbox{\scriptsize.}\hbox{\scriptsize.}}}%
                     =}

 { \begin{list}%
         {$\bullet$}%
         {\setlength{\labelwidth}{20pt}%
          \setlength{\leftmargin}{25pt}%
          \setlength{\topsep}{0pt}
          \setlength{\itemsep}{0pt}
          \setlength{\parsep}{0pt}}}%
 { \end{list} }

\newcommand{\repseq}[3]{\xymatrix @-1pc{#1^{(#2)} \ar@{<->}[r] & \cdots 
\ar@{<->}[r] & #1^{(#3)} }}

\newcommand{\Field}{\mathbb{F}}

\newcommand{\Basis}{\mathcal{B}}

\newcommand{\sprod}{\langle\cdot,\cdot\rangle}

\newcommand{\R}{\mathbb{R}}

\newcommand{\idcat}{\mathds{1}}

\newcommand{\quiv}{\mathcal{Q}}

\newcommand{\Xmod}{\mathbb{X}}

\newcommand{\Umod}{\mathbb{U}}
\newcommand{\Wmod}{\mathbb{W}}

\newcommand{\Vmod}{\mathbb{V}}

\newcommand{\Imod}{\mathbb{I}}

\newcommand{\W}{W} 

 \newcommand{\bidx}{\mathrm{{\bf b}}}
 \newcommand{\didx}{\mathrm{{\bf d}}}
 \newcommand{\leqb}{\leq_{\bidx}}
 \newcommand{\geqb}{\geq_{\bidx}}
 \newcommand{\leqd}{\leq_{\didx}}
 \newcommand{\geqd}{\geq_{\didx}}

\DeclareMathOperator*{\minb}{min_{\leqb}}
\DeclareMathOperator*{\mind}{min_{\leqd}}

\DeclareMathOperator*{\maxb}{max_{\leqb}}
\DeclareMathOperator*{\maxd}{max_{\leqd}}

\newcommand{\Matrix}{\mathbf{M}}

\newcommand{\im}{\mathrm{im} \ }

\newcommand{\V}{V}                        
\newcommand{\K}{\mathbf{K}}               

\newcommand{\Chain}{\mathbf{C}}

\newcommand{\Hom}{\mathbf{H}}

%% file: main.bbl
\begin{thebibliography}{10}

\bibitem{dipha_lib}
Ulrich Bauer, Michael Kerber, and Jan Reininghaus.
\newblock \textsc{DIPHA}.
\newblock https://github.com/DIPHA/dipha.

\bibitem{Bauer:arXiv1303.0477}
Ulrich Bauer, Michael Kerber, and Jan Reininghaus.
\newblock Clear and compress: Computing persistent homology in chunks.
\newblock In {\em Topological Methods in Data Analysis and Visualization III},
  pages 103--117. 2014.

\bibitem{DBLP:conf/alenex/BauerKR14}
Ulrich Bauer, Michael Kerber, and Jan Reininghaus.
\newblock Distributed computation of persistent homology.
\newblock In {\em ALENEX}, pages 31--38, 2014.

\bibitem{boissonnatmariacamalgorithmica}
Jean-Daniel Boissonnat, Tamal~K. Dey, and Cl{\'e}ment Maria.
\newblock The compressed annotation matrix: An efficient data structure for
  computing persistent cohomology.
\newblock {\em Submitted to Algorithmica}, 2014.

\bibitem{boissonnatmariaesa2014}
Jean-Daniel Boissonnat and Cl{\'e}ment Maria.
\newblock Computing persistent homology with various coefficient fields in a
  single pass.
\newblock In {\em ESA}, 2014.

\bibitem{Carlsson:2008:LBS:1325290.1325300}
Gunnar Carlsson, Tigran Ishkhanov, Vin Silva, and Afra Zomorodian.
\newblock On the local behavior of spaces of natural images.
\newblock {\em Int. J. Comput. Vision}, 76:1--12, January 2008.

\bibitem{DBLP:journals/focm/CarlssonS10}
Gunnar~E. Carlsson and Vin de~Silva.
\newblock Zigzag persistence.
\newblock {\em Foundations of Computational Mathematics}, 10(4):367--405, 2010.

\bibitem{DBLP:conf/compgeom/CarlssonSM09}
Gunnar~E. Carlsson, Vin de~Silva, and Dmitriy Morozov.
\newblock Zigzag persistent homology and real-valued functions.
\newblock In {\em Symposium on Computational Geometry}, pages 247--256, 2009.

\bibitem{Chen11persistenthomology}
Chao Chen and Michael Kerber.
\newblock Persistent homology computation with a twist.
\newblock In {\em Proceedings 27th European Workshop on Computational
  Geometry}, 2011.

\bibitem{DBLP:journals/comgeo/ChenK13}
Chao Chen and Michael Kerber.
\newblock An output-sensitive algorithm for persistent homology.
\newblock {\em Comput. Geom.}, 46(4):435--447, 2013.

\bibitem{DBLP:journals/dcg/Cohen-SteinerEH07}
David Cohen-Steiner, Herbert Edelsbrunner, and John Harer.
\newblock Stability of persistence diagrams.
\newblock {\em Discrete {\&} Computational Geometry}, 37(1):103--120, 2007.

\bibitem{DBLP:conf/compgeom/Cohen-SteinerEM06}
David Cohen-Steiner, Herbert Edelsbrunner, and Dmitriy Morozov.
\newblock Vines and vineyards by updating persistence in linear time.
\newblock In {\em Symposium on Computational Geometry}, pages 119--126, 2006.

\bibitem{DBLP:journals/dcg/SilvaMV11}
Vin de~Silva, Dmitriy Morozov, and Mikael Vejdemo-Johansson.
\newblock Persistent cohomology and circular coordinates.
\newblock {\em Discrete {\&} Computational Geometry}, 45(4):737--759, 2011.

\bibitem{DBLP:conf/compgeom/DeyFW14}
Tamal~K. Dey, Fengtao Fan, and Yusu Wang.
\newblock Computing topological persistence for simplicial maps.
\newblock In {\em Symposium on Computational Geometry}, page 345, 2014.

\bibitem{DBLP:conf/esa/DeySY16}
Tamal~K. Dey, Dayu Shi, , and Yusu Wang.
\newblock Simba: An efficient tool for approximating rips-filtration
  persistence via simplicial batch-collapse.
\newblock In {\em Proc. European Symposium on Algorithms}, 2016.

\bibitem{DBLP:books/daglib/0025666}
Herbert Edelsbrunner and John Harer.
\newblock {\em Computational Topology - an Introduction}.
\newblock American Mathematical Society, 2010.

\bibitem{DBLP:journals/dcg/EdelsbrunnerLZ02}
Herbert Edelsbrunner, David Letscher, and Afra Zomorodian.
\newblock Topological persistence and simplification.
\newblock {\em Discrete Comput. Geom.}, 28(4):511--533, 2002.

\bibitem{gabrieltheoremoriginal}
Peter Gabriel.
\newblock Unzerlegbare darstellungen.
\newblock {\em Manuscripta Mathematica}, 1972.

\bibitem{DBLP:conf/isaac/KerberS13}
Michael Kerber and R.~Sharathkumar.
\newblock Approximate {\v{c}}ech complex in low and high dimensions.
\newblock In {\em Algorithms and Computation - 24th International Symposium,
  {ISAAC} 2013, Hong Kong, China, December 16-18, 2013, Proceedings}, pages
  666--676, 2013.

\bibitem{Lichman:2013}
M.~Lichman.
\newblock {UCI} machine learning repository, 2013.

\bibitem{DBLP:journals/eswa/MadeoPL16}
Renata Cristina~Barros Madeo, Sarajane~Marques Peres, and Clodoaldo~Aparecido
  de~Moraes~Lima.
\newblock Gesture phase segmentation using support vector machines.
\newblock {\em Expert Syst. Appl.}, 56:100--115, 2016.

\bibitem{gudhi:PersistentCohomology}
Cl\'ement Maria.
\newblock Persistent cohomology.
\newblock In {\em {GUDHI} User and Reference Manual}. {GUDHI Editorial Board},
  2015.

\bibitem{DBLP:conf/soda/MariaO15}
Cl{\'{e}}ment Maria and Steve~Y. Oudot.
\newblock Zigzag persistence via reflections and transpositions.
\newblock In Piotr Indyk, editor, {\em Proceedings of the Twenty-Sixth Annual
  {ACM-SIAM} Symposium on Discrete Algorithms, {SODA} 2015, San Diego, CA, USA,
  January 4-6, 2015}, pages 181--199. {SIAM}, 2015.

\bibitem{DBLP:conf/compgeom/MilosavljevicMS11}
Nikola Milosavljevic, Dmitriy Morozov, and Primoz Skraba.
\newblock Zigzag persistent homology in matrix multiplication time.
\newblock In {\em Symposium on Comp. Geom.}, 2011.

\bibitem{dionysus_morozov}
Dmitriy Morozov.
\newblock {\tt Dionysus}.
\newblock http://www.mrzv.org/software/dionysus/.

\bibitem{Munkres-elementsalgtop1984}
James~R. Munkres.
\newblock {\em Elements of algebraic topology}.
\newblock Addison-Wesley, 1984.

\bibitem{DBLP:journals/corr/OtterPTGH15}
Nina Otter, Mason~A. Porter, Ulrike Tillmann, Peter Grindrod, and Heather~A.
  Harrington.
\newblock A roadmap for the computation of persistent homology.
\newblock {\em CoRR}, abs/1506.08903, 2015.

\bibitem{os-zz-14}
Steve~Y. Oudot and Donald~R. Sheehy.
\newblock {Z}igzag {Z}oology: {R}ips {Z}igzags for {H}omology {I}nference.
\newblock {\em J. Foundations of Computational Mathematics}, 15(5):1151--1186,
  2015.

\bibitem{DBLP:journals/dcg/Sheehy13}
Donald Sheehy.
\newblock Linear-size approximations to the vietoris-rips filtration.
\newblock {\em Discrete {\&} Computational Geometry}, 49(4):778--796, 2013.

\bibitem{simba_lib}
Dayu Shi.
\newblock \textsc{SimBa}.
\newblock http://web.cse.ohio-state.edu/~tamaldey/SimBa/Simba.html.

\bibitem{DBLP:journals/dcg/ZomorodianC05}
Afra Zomorodian and Gunnar~E. Carlsson.
\newblock Computing persistent homology.
\newblock {\em Discrete {\&} Computational Geometry}, 33(2):249--274, 2005.

\end{thebibliography}
